\documentclass[a4paper,USenglish,thm-restate,cleveref,numberwithinsect]{lipics-v2021}
\pdfoutput=1
\hideLIPIcs
\nolinenumbers

\title{Classification of Local Optimization Problems in Directed Cycles}
\author{Thomas Boudier}{GSSI}{thomas.boudier@gssi.it}{https://orcid.org/0009-0002-0119-9333}{}
\author{Fabian Kuhn}{University of Freiburg}{}{}{}
\author{Augusto Modanese}{CISPA Helmholtz Center for Information Security}{augusto.modanese@cispa.de}{https://orcid.org/0000-0003-0518-8754}{Most of this work done while affiliated with Aalto University}
\author{Ronja Stimpert}{Aalto University}{ronja.stimpert@aalto.fi}{https://orcid.org/0009-0000-5794-0034}{}
\author{Jukka Suomela}{Aalto University}{jukka.suomela@aalto.fi}{https://orcid.org/0000-0001-6117-8089}{}
\authorrunning{T. Boudier, F. Kuhn, A. Modanese, R. Stimpert, and J. Suomela}
\Copyright{Thomas Boudier, Fabian Kuhn, Augusto Modanese, Ronja Stimpert, and Jukka Suomela}
\funding{This work was supported in part by the Research Council of Finland, Grant 359104, and by the Quantum Doctoral Education Pilot, the Ministry of Education and Culture, decision VN/3137/2024-OKM-4.}
\acknowledgements{This project was initiated at the Research Workshop on Distributed Algorithms (RW-DIST 2025) in Freiburg, Germany; we would like to thank all workshop participants and organizers for inspiring discussions.}

\usepackage{amsmath}
\usepackage{amsthm}
\usepackage{amssymb}
\usepackage{complexity}
\usepackage{cite}
\usepackage{booktabs}
\usepackage{enumitem}
\usepackage{csquotes}
\usepackage{mathtools,etoolbox}
\usepackage{mathdots}
\usepackage{xcolor}
\usepackage{amsfonts}
\usepackage{local-models}

\let\R\relax
\usepackage{ammacros}

\newcommand{\bopt}{{\beta_{\operatorname{opt}}}}
\newcommand{\bflex}{{\beta_{\operatorname{flex}}}}
\newcommand{\bcoprime}{{\beta_{\operatorname{coprime}}}}
\newcommand{\bgap}{{\beta_{\operatorname{gap}}}}
\newcommand{\bconst}{{\beta_{\operatorname{const}}}}
\newcommand{\Gopt}{{G_{\operatorname{opt}}}}
\newcommand{\Gflex}{{G_{\operatorname{flex}}}}

\newcommand{\Ggap}{{G_{\operatorname{gap}}}}
\newcommand{\Gconst}{{G_{\operatorname{const}}}}
\newcommand{\dflex}{{\delta_{\operatorname{flex}}}}
\newcommand{\dgap}{{\delta_{\operatorname{gap}}}}

\DeclareMathOperator*{\obj}{obj}
\DeclareMathOperator*{\aggr}{aggr}
\DeclareMathOperator*{\opt}{OPT}
\DeclareMathOperator{\expectation}{E}

\usepackage{amssymb,stackengine,scalerel}

\begin{document}

\maketitle

\begin{abstract}
We present a complete classification of the distributed computational complexity of local optimization problems in directed cycles for both the deterministic and the randomized LOCAL model. We show that for any local optimization problem $\Pi$ (that can be of the form min-sum, max-sum, min-max, or max-min, for any local cost or utility function over some finite alphabet), and for any \emph{constant} approximation ratio $\alpha$, the task of finding an $\alpha$-approximation of $\Pi$ in directed cycles has one of the following complexities:
  \begin{enumerate}
    \item $O(1)$ rounds in deterministic LOCAL, $O(1)$ rounds in randomized LOCAL,
    \item $\Theta(\log^* n)$ rounds in deterministic LOCAL, $O(1)$ rounds in randomized LOCAL,
    \item $\Theta(\log^* n)$ rounds in deterministic LOCAL, $\Theta(\log^* n)$ rounds in randomized LOCAL,
    \item $\Theta(n)$ rounds in deterministic LOCAL, $\Theta(n)$ rounds in randomized LOCAL.
  \end{enumerate}
Moreover, for any given $\Pi$ and $\alpha$, we can determine the complexity class automatically, with an efficient (centralized, sequential) meta-algorithm, and we can also efficiently synthesize an asymptotically optimal distributed algorithm.

Before this work, similar results were only known for local search problems (e.g., locally checkable labeling problems). The family of local optimization problems is a strict generalization of local search problems, and it contains numerous commonly studied distributed tasks, such as the problems of finding approximations of the maximum independent set, minimum vertex cover, minimum dominating set, and minimum vertex coloring.
\end{abstract}

\thispagestyle{empty}
\clearpage
\pagenumbering{arabic}

\section{Introduction}

In this work we study distributed algorithms for \textbf{local optimization problems}---these include a wide range of familiar graph optimization problems such as \emph{maximum independent set}, \emph{minimum vertex cover}, \emph{minimum dominating set}, and \emph{minimum vertex coloring}. We show that in \textbf{directed cycles}, essentially all questions about the existence of distributed algorithms for local optimization problems can be answered in a mechanical and efficient manner. We give a \textbf{meta-algorithm} that answers the following questions, for any given local optimization problem $\Pi$ and any given approximation factor $\alpha$:
\begin{enumerate}
  \item How fast can we find an $\alpha$-approximation of $\Pi$ in the \detlcl model?
  \item How fast can we find an $\alpha$-approximation of $\Pi$ in the \randlcl model?
\end{enumerate}

\subsection{Broader Context}

There is a long line of prior work that has studied \textbf{local search problems}, and in particular \emph{locally checkable labeling problems} or LCLs; these are problems where the task is merely to satisfy all local constraints, i.e., the task is to find a feasible solution that does not need to be optimal in any sense. Familiar examples include the problems of finding a \emph{maximal} independent set, or a valid coloring with $\Delta+1$ colors in a graph of maximum degree $\Delta$.

While a lot of early work focused on specific concrete examples of such tasks, e.g., exactly what is the distributed complexity of finding a $3$-coloring of a cycle \cite{cole-vishkin-1986-deterministic-coin-tossing-with,linial-1992-locality-in-distributed-graph-algorithms}, Naor and Stockmeyer \cite{naor-stockmeyer-1995-what-can-be-computed-locally} initiated the study of \textbf{meta-computational questions} in this context: for any given local search problem $\Pi$, can we systematically determine what the distributed computational complexity of $\Pi$ is?

It turns out that such questions are in general \emph{undecidable}, already in the case of seemingly simple graph classes such as 2-dimensional grids \cite{naor-stockmeyer-1995-what-can-be-computed-locally,brandt-hirvonen-etal-2017-lcl-problems-on-grids}. However, in graph families such as \textbf{cycles, paths, and trees} many such questions are decidable \cite{chang-studeny-suomela-2023-distributed-graph-problems,balliu-brandt-etal-2022-efficient-classification-of,balliu-brandt-etal-2023-locally-checkable-problems-in,balliu-brandt-etal-2019-the-distributed-complexity-of}. Furthermore, while such questions are \emph{computationally hard} if we have input labels \cite{balliu-brandt-etal-2019-the-distributed-complexity-of}, many of these questions can be solved efficiently (in time that is polynomial in the size of the description of problem~$\Pi$) in the case of \textbf{unlabeled} graphs \cite{chang-studeny-suomela-2023-distributed-graph-problems,balliu-brandt-etal-2022-efficient-classification-of}.

In particular, if we look at LCL problems on unlabeled directed cycles, they are very well understood: the round complexity of any such problem is $O(1)$, $\Theta(\log^* n)$, or $\Theta(n)$ in all of our usual models of distributed computing (\detlcl, \randlcl, deterministic \congest, and randomized \congest), and there is a computer program that automatically finds out the complexity class for any given problem \cite{tereshchenko-studeny-2021-lcl-classifier}.

However, many problems that have been studied in the field of distributed graph algorithms are \textbf{local optimization problems}, e.g., \cite{kuhn-moscibroda-wattenhofer-2016-local-computation,czygrinow-hanckowiak-wawrzyniak-2008-fast-distributed,goos-suomela-2014-no-sublogarithmic-time-approximation,goos-hirvonen-suomela-2013-lower-bounds-for-local,lenzen-wattenhofer-2008-leveraging-linial-s-locality,heydt-kublenz-etal-2025-distributed-domination-on,harris-2020-distributed-local-approximation-algorithms}, and much less is currently known about them. There is a large collection of isolated examples that apply to specific problems, but there is no complete classification of possible distributed round complexities of local optimization problems, and there is no meta-algorithm that is able to automatically determine the complexity of a given problem.

Furthermore, what makes the question particularly intriguing is that it is known that the complexity landscape has to be fundamentally different, already in the simplest case of unlabeled directed cycles. To see this, consider the problem of finding a $1.1$-approximation of a minimum dominating set in cycles. Thanks to \cite{czygrinow-hanckowiak-wawrzyniak-2008-fast-distributed}, it is known that in the \detlcl model, the complexity of this problem is $\Theta(\log^* n)$ rounds, while in the \randlcl model, the complexity is $O(1)$ rounds. The same never happens with any LCL problem: for them, $O(1)$ rounds in \randlcl always implies $O(1)$ rounds in \detlcl \cite{chang-kopelowitz-pettie-2019-an-exponential-separation,naor-stockmeyer-1995-what-can-be-computed-locally}.

Hence we have three questions that we aim to answer in this work:
\begin{enumerate}
  \item What is the right framework for formalizing and representing local optimization problems (analogous to LCL problems in the context of search problems)?
  \item What does the complexity landscape look like for local optimization, especially when we compare \detlcl with \randlcl?
  \item Can we design a meta-algorithm that automatically determines the complexity class for any given local optimization problem and any given approximation ratio?
\end{enumerate}

\subsection{Formalism}\label{ssec:intro-formalism}

In this work, we look at the simplest nontrivial case: local optimization problems in directed cycles. Let us first recall how we can specify LCL problems in this setting: There is some finite set of output labels $\Gamma$ and verification distance $r$, and we specify which $(r+1)$-tuples of output labels are valid. For example, for the task of $3$-coloring, we have $\Gamma = \{1,2,3\}$, and we can select $r=1$, and the set of valid tuples of output labels is $\{12, 13, 21, 23, 31, 32\}$. A feasible solution is a labeling $s \colon V \to \Gamma$ that is everywhere locally valid: for each node $v$, the sequence of output labels $s(N_r(v))$ formed by $v$ and its $r$ successors is a valid tuple.

\begin{remark*}
  In cycles, it does not really matter whether we look at $v$ and its $r$ successors, or $v$ and its $r$ predecessors, or $v$ and its radius-$r/2$ symmetric neighborhood. In this work we follow the convention that a problem is defined in terms of $r$ successors.
\end{remark*}

In local optimization problems we keep the same basic setup, but now we associate a value $c(x)$ with each valid $(r+1)$-tuple $x$. This way, each labeling $s$ associates a \emph{local value} $c(s(N_r(v)))$ with each node (we will use terms \emph{cost} or \emph{utility} as appropriate to refer to the value, depending on the problem). The optimization task is one of the following, combining an \emph{objective function} $\obj \in \{\max,\min\}$ and an \emph{aggregation function} $\aggr \in \{\max,\min,\sum\}$:
\begin{enumerate}
  \item $\max \min$: maximize the minimum of local values,
  \item $\max \sum$: maximize the sum of local values,
  \item $\min \max$: minimize the maximum of local values,
  \item $\min \sum$: minimize the sum of local values.
\end{enumerate}
We will use the value $\bot$ to indicate forbidden label combinations (which can also be interpreted as e.g.\ infinite cost). We are interested in approximations, and we follow the convention that approximation ratio $\alpha$ is always at least $1$, so e.g.\ a $2$-approximation for a minimization problem is a feasible solution with total cost at most twice the optimum, and a $2$-approximation for a maximization problem is a feasible solution with total utility at least half the optimum.

\subsection{Examples of Problems}\label{ssec:intro-examples}

The following examples will hopefully clarify the formalism, and also demonstrate that it is broadly applicable. We will start with familiar, classic graph optimization problems.

\begin{example}[maximum independent set]\label{ex:max-ind-set}
  The maximum independent set problem can be encoded as follows, with $1$ indicating nodes that form the independent set:
  \begin{itemize}
  \item $\Gamma = \{0,1\}$, $r = 1$, $\obj = \max$, $\aggr = \sum$,
  \item $c(00) = 0$,
        $c(01) = 0$,
        $c(10) = 1$,
        $c(11) = \bot$.
  \end{itemize}
  Here any valid independent set has a nonnegative total value, and the value is equal to the size of the independent set.
\end{example}

\begin{example}[minimum dominating set]\label{ex:min-dom-set}
  The minimum dominating set problem can be encoded as follows:
  \begin{itemize}
  \item $\Gamma = \{0,1\}$, $r = 2$, $\obj = \min$, $\aggr = \sum$,
  \item $c(000) = \bot$,
  $c(001) = 0$,
  $c(010) = 0$,
  $c(011) = 0$,\\
  $c(100) = 1$,
  $c(101) = 1$,
  $c(110) = 1$,
  $c(111) = 1$.
  \end{itemize}
\end{example}

\begin{example}[minimum vertex coloring]\label{ex:min-ver-col}
  To encode the minimum vertex coloring problem, we can use e.g.\ the following formulation:
  \begin{itemize}
  \item $\Gamma = \{1,2,3\}$, $r = 1$, $\obj = \min$, $\aggr = \max$,
  \item $c(ij) = i$ if $i \ne j$,
        $c(ii) = \bot$.
  \end{itemize}
  Here for example a valid $2$-coloring (using only colors $1$ and $2$) has total cost $2$, and a valid $3$-coloring has total cost $3$.
\end{example}

\begin{example}[maximum domatic partition]\label{ex:max-dom-par}
  Let us now look at a slightly more interesting example: domatic partition. In this problem, the task is to partition the set of nodes into as many disjoint dominating sets as possible; in a cycle we can have at most $3$ such sets. While more clever ad-hoc encodings are possible, we will use the following approach that generalizes to a wide range of similar problems. We use labels of the form $x_i$, where $i$ indicates that we claim to have a solution with exactly $i$ dominating sets, and $x$ indicates that this node claims to be in a dominating set with label $x$. We then need to ensure that nodes also agree on the number of dominating sets, and each dominating set is also valid. This will suffice:
  \begin{itemize}
    \item $\Gamma = \{a_1,a_2,b_2,a_3,b_3,c_3\}$, $r = 2$, $\obj = \max$, $\aggr = \min$,
    \item $c(a_{1}a_{1}a_{1}) = 1$,
    \item $c(a_{2}a_{2}b_{2}) = c(a_{2}b_{2}a_{2}) = c(a_{2}b_{2}b_{2}) = c(b_{2}a_{2}a_{2}) = c(b_{2}a_{2}b_{2}) = c(b_{2}b_{2}a_{2}) = 2$,
    \item $c(a_{3}b_{3}c_{3}) = c(b_{3}c_{3}a_{3}) = c(c_{3}a_{3}b_{3}) = c(a_{3}c_{3}b_{3}) = c(c_{3}b_{3}a_{3}) = c(b_{3}a_{3}c_{3}) = 3$,
    \item $c(x) = \bot$ otherwise.
  \end{itemize}
\end{example}

\subsection{A Classification Example}\label{ssec:intro-sloppy}

So far we have seen examples of natural $\max \sum$, $\min \sum$, $\max \min$, and $\min \max$ problems. We will introduce a somewhat artificial problem, which is engineered to illustrate many key phenomena in a single, concise package.

\begin{example}[sloppy coloring]\label{ex:sloppy-col}
  In this problem, we can produce a proper $2$-coloring with colors $\{b,w\}$, a proper $3$-coloring with colors $\{1,2,3\}$, or a somewhat sloppy $3$-coloring with colors $\{a,b,c\}$. The problem is defined as follows:
  \begin{itemize}
    \item $\Gamma = \{b,w,1,2,3,a,b,c\}$, $r = 1$, $\obj = \min$, $\aggr = \sum$,
    \item $c(bw) = c(wb) = 1$,
    \item $c(12) = c(21) = c(23) = c(32) = c(13) = c(31) = 2$,
    \item $c(ab) = c(ba) = c(bc) = c(cb) = c(ac) = c(ca) = 3$,
    \item $c(aa) = 100$,
    \item $c(x) = \bot$ otherwise.
  \end{itemize}
\end{example}

Let us try to gain a bit more intuition on the problem first. If we manage to produce a proper $2$-coloring, the total cost in an $n$-cycle is $n$, and if we manage to produce a proper $3$-coloring using labels $\{1,2,3\}$, the total cost is $2n$. If we resort to labels $\{a,b,c\}$ and produce a proper $3$-coloring, the total cost is $3n$. If a sublinear number of edges uses the color pair $aa$, the total cost will be $3n + o(n)$, and if we produce a trivial solution with all nodes colored $a$, the total cost will be $100n$.

Now the key question is this: given some constant $\alpha \ge 1$, what is the distributed round complexity of finding an $\alpha$-approximation of a sloppy coloring? It turns out that the complete answer in this case looks like this:
\begin{itemize}
  \item For $1 \le \alpha < 2$, the complexity is $\Theta(n)$ in both the \detlcl and \randlcl models. The optimal algorithm is a trivial brute-force algorithm that finds a $2$-coloring if it exists, and otherwise produces a $3$-coloring. There is no $o(n)$-round algorithm for finding a $1.999$-approximation or better.
  \item For $2 \le \alpha \le 3$, the complexity is $\Theta(\log^* n)$ in both the \detlcl and \randlcl models. We can always produce a $3$-coloring, using the Cole--Vishkin algorithm \cite{cole-vishkin-1986-deterministic-coin-tossing-with}, and this is always a $2$-approximation. There is no $o(\log^* n)$-round deterministic or randomized algorithm for finding a $3$-approximation or better.
  \item For $3 < \alpha < 100$, the complexity is $\Theta(\log^* n)$ in the \detlcl model but $O(1)$ in the \randlcl model. The randomized algorithm essentially finds a random ruling set to divide the cycle in segments. The algorithm uses color $a$ at the segment boundaries and fills in all sufficiently short segments with $b,c,b,c,\dotsc$. Super-constant-length segments are rare, and there we can use the more costly solution $a,a,a,a,\dotsc$. By adjusting the parameters, we can achieve approximation ratio $3+\epsilon$ for any constant $\epsilon > 0$ with high probability (the worst case being that the input is an even cycle that could be $2$-colored). There is no $o(\log^* n)$-round deterministic algorithm for finding a $99.999$-approximation or better.
  \item For $\alpha \ge 100$, the complexity is $O(1)$ in both the \detlcl and \randlcl models. The algorithm simply outputs $a$ everywhere.
\end{itemize}
As we will see, this turns out to be a well-representative problem: for \emph{any} local optimization problem in our problem family, there are at most these four complexity classes, and the optimal algorithms are also essentially generalizations of the above four algorithms. We can automatically deduce all relevant $\alpha$ thresholds between the classes.

\subsection{Contributions}

We show that for any local optimization problem $\Pi$ (in the formalism of \cref{ssec:intro-formalism}), and for any constant $\alpha$, the task of finding an $\alpha$-approximation of $\Pi$ in directed cycles has one of these complexities:
\begin{enumerate}
  \item $O(1)$ rounds in \detlcl, $O(1)$ rounds in \randlcl,
  \item $\Theta(\log^* n)$ rounds in \detlcl, $O(1)$ rounds in \randlcl,
  \item $\Theta(\log^* n)$ rounds in \detlcl, $\Theta(\log^* n)$ rounds in \randlcl,
  \item $\Theta(n)$ rounds in \detlcl, $\Theta(n)$ rounds in \randlcl.
\end{enumerate}
Moreover, for any given $\Pi$ and $\alpha$, we can determine the complexity class automatically, with an efficient (centralized, sequential) meta-algorithm, and we can also efficiently synthesize an asymptotically optimal distributed algorithm. In particular, all local optimization problems similar to those in \cref{ssec:intro-examples,ssec:intro-sloppy} are now fully understood in the case of directed cycles.

One key consequence of our work is that increasing the running time from
$\Theta(\log^* n)$ to, say, $\Theta(\log n)$ or $\Theta(\sqrt{n})$ does not give any
\emph{constant} advantage in the approximation ratio, for any local optimization
problem in directed cycles. 
If we are interested in finding a $(1+1/n)$-approximation, for example, then the
situation is more complex, and indeed we might still obtain some sub-constant
advantage by doing so.
However, as we show, we will never be able to improve our approximation ratio
from, say, $1.2345$ to $1.2344$ by such means. You will need to jump from $\Theta(\log^* n)$ rounds all the way to $\Theta(n)$ rounds for such an improvement.

We emphasize that the case of max-min and min-max problems can be handled by known techniques, as each fixed $\alpha$ essentially defines a sub-problem that is an LCL. The main novel part is $\min\sum$ and $\max\sum$ problems, which require genuinely new techniques and ideas. However, we wanted to give a unified formalization of all of these problems, and show that the same definitions and concepts indeed largely apply to both problem families.

\subsection{Techniques}

We take a three-step approach:

\subparagraph*{Step 1: Defining the problem parameters (\cref{sec:problem-param}).}

We first \emph{define} seven parameters $\bopt$, $\bflex$, $\dflex$, $\bcoprime$, $\bgap$, $\dgap$, and $\bconst$ that capture essentially everything that we need to know about a given optimization problem~$\Pi$.

To do this, we define a node-weighted de Bruijn graph $G$ that captures the entire structure of~$\Pi$, including the values. The intuition is that for any input cycle $C$, a feasible solution of~$\Pi$ in $C$ corresponds to a closed walk in $G$, with the same value. Then we proceed to define four subgraphs of $G$:
\begin{enumerate}
  \item We let $\Gopt = G$. We define $\bopt$ based on the cheapest closed walk in $\Gopt$.
  \item We zoom into \emph{flexible components} $\Gflex \subseteq \Gopt$, in the spirit of \cite{chang-studeny-suomela-2023-distributed-graph-problems}. In brief, a node $v$ of $G$ is \emph{$K$-flexible} if there are self-returning walks from $v$ back to $v$ of all possible lengths $K, K+1, K+2, \dotsc$. A node is \emph{flexible} if it is $K$-flexible for some $K$. It follows that in each strongly connected component either all or none of the nodes are flexible, and we say that a component is flexible if the nodes are flexible. We define parameters $\bflex$ and $\dflex$ based on the cheapest closed walks in $\Gflex$.
  \item We zoom into \emph{flexible components with self-loops} $\Ggap \subseteq \Gflex$; we simply take those flexible components that also contain nodes with a self-loop. We define parameters $\bgap$ and $\dgap$ based on the cheapest closed walks in $\Ggap$.
  \item We form the subgraph that consists only of \emph{self-loops} $\Gconst \subseteq \Ggap$. We define parameter $\bconst$ based on the cheapest closed walks (that is, cheapest self-loops) in $\Gconst$.
\end{enumerate}
The above six parameters suffice for $\min\sum$ and $\max\sum$ problems. For max-min and min-max problems we also need to introduce a parameter $\bcoprime$ that, in brief, captures the cheapest pair of closed walks that have coprime lengths.

\begin{table}
  \caption{All $\beta$ values defined in \cref{sec:problem-param} for the examples in \cref{ssec:intro-examples}. Missing values are not essential to capture the problem's classification, regardless of whether the value can be obtained or not.}
  \label{table:strategies-ex}
  \centering
  \begin{tabular}{llllllll}
    \toprule
    Example & $\bopt$ & $\bflex$ & $\dflex$ & $\bcoprime$ & $\bgap$ & $\dgap$ & $\bconst$\\
    \midrule
    \cref{ex:max-ind-set}: max.~independent set & $\frac{1}{2}$ & $\frac{1}{2}$ & true & - & $\frac{1}{2}$ & true & $0$\\
    \cref{ex:min-dom-set}: min.~dominating set & $\frac{1}{3}$ & $\frac{1}{3}$ & true & - & $\frac{1}{3}$ & true & $1$\\
    \cref{ex:min-ver-col}: min.~vertex coloring & $2$ & - & false & $3$ & $\bot$ & false & $\bot$\\
    \cref{ex:max-dom-par}: max.~domatic partition & $3$ & - & false & $2$ & $1$ & false & $1$\\
    \cref{ex:sloppy-col}: sloppy coloring & $1$ & $2$ & false & - & $3$ & true & $100$\\
    \bottomrule
  \end{tabular}
\end{table}

\Cref{table:strategies-ex} shows concrete examples of these parameter values for the examples from \cref{ssec:intro-examples}. We emphasize that these parameters have a purely graph-theoretic definition; their definitions do not make any references to e.g.\ a particular model of distributed computing. However, in step 3 we will see how these parameters connect with the task at hand.

\subparagraph*{Step 2: The problem parameters are efficiently computable (\cref{sec:efficient}).}

We show that given a description of any local optimization problem $\Pi$, we can efficiently compute all necessary parameters out of $\bopt$, $\bflex$, $\dflex$, $\bcoprime$, $\bgap$, $\dgap$, and $\bconst$. In each case, we first construct the de Bruijn graph $G$, which has size polynomial in the size of the problem description. Then we show that we can efficiently construct the relevant subgraph ($\Gopt, \Gflex, \Ggap, \Gconst$). Finally, we show that we can design an efficient graph algorithm for determining the corresponding parameter values based on the structure of the subgraph.

For example, to determine the value of $\bopt$ for a $\min\sum$ problem, it is sufficient to find a cycle in $\Gopt = G$ that minimizes the \emph{average} cost; we show that all of this can be indeed done efficiently, largely relying on standard graph algorithms.

Determining e.g.\ $\bflex$ and $\dflex$ is more interesting, and there we will build on and extend prior work related to flexibility \cite{chang-studeny-suomela-2023-distributed-graph-problems}. In brief, it is no longer sufficient to find a single cheap cycle, but we also need to find a \emph{pair of low-average-cost cycles with coprime lengths}. We refer to \cref{sec:efficient} for the details of how all seven parameters can be determined efficiently.

\subparagraph*{Step 3: The problem parameters determine distributed complexity (\cref{sec:alg}).}

Finally, we show that the complexity of any given problem for any approximation
ratio $\alpha$ can be determined using parameters $\bopt$, $\bflex$, $\dflex$, $\bcoprime$, $\bgap$, $\dgap$, and $\bconst$. Furthermore, we can also automatically construct an asymptotically-fastest algorithm for any given $\alpha$.

The intuition here is that $\bopt$ provides a baseline against which we can compare. For example, for $\min\sum$ problems, for infinitely many values of $n$, there exists a solution of total cost $\bopt n$, and for no value of $n$ we have cheaper solutions. So an average cost of $\bopt$ per node represents the best-case scenario for the optimal solution. In particular, if we can find a solution of cost at most $\alpha \bopt n$ for all $n$, it will be clearly an $\alpha$-approximation.

For a full classification, we will now need both upper bounds and matching lower bounds. Let us first consider upper bounds. It is perhaps easier to consider what we can do with a given time budget, using $\min\sum$ problems as an example:
\begin{enumerate}
  \item If we can use $O(n)$ rounds, we can brute-force an optimal solution, and hence trivially find a $1$-approximation, regardless of the input size.
  \item If we can use only $O(\log^* n)$ rounds, we can use the following idea: find a sufficiently-sparse ruling set that divides the cycle into segments, so that the length of each segment is at least $A$ and at most $B$ for some appropriately-chosen constants $A$ and $B$. Then we resort to the notion of flexibility: we put some fixed ``flexible output sequence'' $x$ at segment boundaries. Then we can follow a cheap closed walk from $x$ to $x$ in the de Bruijn graph $\Gflex$ to fill the gap between two such markers. There are two cases:
  \begin{enumerate}
    \item There are two closed walks from $x$ to $x$ in $\Gflex$ that have coprime lengths and that have average cost exactly $\bflex$. Then by using a combination of such walks we can fill in gaps so that the total cost of the solution is exactly $\bflex n$, and hence we can achieve an approximation ratio of $\alpha = \bflex/\bopt$.
    \item Otherwise we will occasionally use slightly more expensive fragments, but we can still achieve a total cost of $\bflex n + o(n)$, and hence approximation ratio $\alpha = \bflex/\bopt + \epsilon$ for any constant $\epsilon > 0$.
  \end{enumerate}
  \item If we further restrict the running time to $O(1)$ rounds, but allow the use of randomness, we can use a faster but sloppier approach for dividing the cycle into segments by choosing a random ruling set. The key difference is that now we can no longer guarantee that all segments have length at most some constant $B$, yet we need to be able to solve the problem in constant time. This is where we resort to self-loops in $\Ggap$: in those (rare) segments that are too long, we simply fill in the middle part of the segment by following a self-loop, that is, we produce in those areas a constant solution. For example, in the maximum independent set problem this means that we may occasionally have long stretches of $0$s in our output. By choosing the parameters appropriately, we can achieve a total cost of $\bgap n + o(n)$, and hence an approximation ratio of $\alpha = \bgap/\bopt + \epsilon$.
  \item Finally, if we must use an $O(1)$-round deterministic algorithm, we will produce a constant output that corresponds to following the cheapest self-loop of $\Gconst$ repeatedly, achieving a total cost of $\bconst n$ and hence approximation ratio of $\alpha = \bconst/\bopt$.
\end{enumerate}

Then we need to show that the above strategy is optimal: there are no other complexity classes, and for each class the above approximation ratios are the best possible. To give some spirit of the lower bounds, we for example reason as follows. Assume that we have an algorithm that runs in $o(n)$ rounds. Then if the algorithm succeeds in producing a valid solution for every instance, it cannot use output values outside $\Gflex$. Hence the output of the algorithm forms a closed walk in $\Gflex$. But such walks have an average cost at least $\bflex$, while the optimum might have an average cost of $\bopt$, and it follows that the algorithm cannot find an approximation better than $\alpha = \bflex/\bopt$.

The lower bounds get more technical in the $O(\log^* n)$ region. There we use Naor's randomized version
\cite{naor-1991-a-lower-bound-on-probabilistic-algorithms-for} of Linial's lower
bound \cite{linial-1992-locality-in-distributed-graph-algorithms} as well as the
Ramsey-theoretic lower bound technique familiar from Czygrinow, Hanckowiak, and
Wawrzyniak \cite{czygrinow-hanckowiak-wawrzyniak-2008-fast-distributed}. We
refer to \cref{sec:alg} for the details.

\subsection{Full Classification}

\Cref{table:classes-to-params-sum,table:classes-to-params-minmax} summarize how we can determine the exact complexity class for any given problem $\Pi$ and any given approximation ratio $\alpha$. The classification is complete; there are matching upper and lower bounds for each case. \Cref{table:classification-ex} shows what is obtained from applying our method to the previous examples from \cref{ssec:intro-examples}.

\begin{table}
  \centering 
  \caption{Scheme for inferring the complexity class from the problem parameters
  of sum problems.
  To determine the optimal complexity for a problem and a given $\alpha$, one
  scans the \enquote{sufficient condition} column until a line is matched.
  Notice that the bounds given are tight in the sense that, if a condition on
  the right is satisfied and none of those above it apply, then the complexity
  on the left is also a lower bound for the corresponding approximation problem.
  For example, if we have a min $\sum$ problem where $\alpha\bopt \ge \bgap$ and
  $\dgap$ is false but $\alpha\bopt < \bconst$ holds, then there is an
  $\Omega(\log^* n)$ lower bound for \detlcl.}
  \begin{tabular}{llll}
    \toprule
    \multicolumn{2}{c}{Complexity} & \multicolumn{2}{c}{Sufficient condition} \\
    \cmidrule(r){1-2} \cmidrule(r){3-4}
    det.~\local & rand.~\local & $\min \sum$ & $\max \sum$\\
    \midrule
    $O(1)$ & $O(1)$ 
    & $\alpha\bopt \geq \bconst$ & $\frac{1}{\alpha}\bopt \leq \bconst$ \\ 
    \midrule
    $\Theta(\log^{*} n)$ & $O(1)$ 
    & $\alpha\bopt \geq \bgap$ and $\dgap$ is false 
    & $\frac{1}{\alpha}\bopt \leq \bgap$ and $\dgap$ is false \\
    & & $\alpha\bopt > \bgap$ and $\dgap$ is true 
    & $\frac{1}{\alpha}\bopt < \bgap$ and $\dgap$ is true \\
    \midrule
    $\Theta(\log^{*} n)$ & $\Theta(\log^{*} n)$ 
    & $\alpha\bopt \geq \bflex$ and $\dflex$ is false 
    & $\frac{1}{\alpha}\bopt \leq \bflex$ and $\dflex$ is false\\
    & & $\alpha\bopt > \bflex$ and $\dflex$ is true 
    & $\frac{1}{\alpha}\bopt < \bflex$ and $\dflex$ is true\\
    \midrule
    $\Theta(n)$ & $\Theta(n)$ 
    & $\alpha\bopt < \bflex$ 
    & $\frac{1}{\alpha}\bopt > \bflex$\\
    \midrule
    unsolvable & unsolvable & $\bflex = \bot$ & $\bflex = \bot$ \\
    \bottomrule
  \end{tabular}
  \label{table:classes-to-params-sum}
\end{table}

\begin{table}
  \centering 
  \caption{Scheme for deriving the complexity class from the problem parameters
  of min-max and max-min problems. 
  The same remarks regarding determining the complexity of a problem as in
  \cref{table:classes-to-params-sum} apply.
  Notice that, unlike the case of sum problems, there is not the possibility of
  a deterministic $\Theta(\log^* n)$ and randomized $O(1)$ complexity
  class here.}
  \begin{tabular}{llll}
    \toprule
    \multicolumn{2}{c}{Complexity} & \multicolumn{2}{c}{Sufficient condition} \\
    \cmidrule(r){1-2} \cmidrule(r){3-4}
    det.~\local & rand.~\local & $\min \max$ & $\max \min$\\
    \midrule
    $O(1)$ & $O(1)$ 
    & $\alpha\bopt \geq \bconst$ & $\frac{1}{\alpha}\bopt \leq \bconst$ \\ 
    \midrule
    $\Theta(\log^{*} n)$ & $\Theta(\log^{*} n)$ 
    & $\alpha\bopt \geq \bcoprime$ 
    & $\frac{1}{\alpha}\bopt \leq \bcoprime$\\
    \midrule
    $\Theta(n)$ & $\Theta(n)$ 
    & $\alpha\bopt < \bcoprime$ 
    & $\frac{1}{\alpha}\bopt > \bcoprime$\\
    \midrule
    unsolvable & unsolvable & $\bcoprime = \bot$ & $\bcoprime = \bot$ \\
    \bottomrule
  \end{tabular}
  \label{table:classes-to-params-minmax}
\end{table}

\begin{table}
  \centering 
  \caption{Complexity of computing an $\alpha$-approximation in the
  deterministic and randomized \local models for each of our examples from
  \cref{ssec:intro-examples}.}
  \begin{tabular}{lllll}
    \toprule
    &&& \multicolumn{2}{c}{Complexity} \\
    \cmidrule(r){4-5}
    Example & Range of $\alpha$ & Strategy & det.~\local & rand.~\local \\
    \midrule
    \Cref{ex:max-ind-set}: & $1$ & Optimal & $\Theta(n)$ & $\Theta(n)$ \\ 
    max~independent set & $(1,\infty)$ & Constant fragment & $\Theta(\log^{*} n)$ & $O(1)$ \\
    \midrule
    \Cref{ex:min-dom-set}: & $1$ & Optimal & $\Theta(n)$ & $\Theta(n)$ \\ 
    min.~dominating set  & $(1,3)$ & Constant fragment & $\Theta(\log^{*} n)$ & $O(1)$ \\ 
    & $[3,\infty)$ & Constant solution & $O(1)$ & $O(1)$ \\
    \midrule
    \Cref{ex:min-ver-col}: & $[1,\tfrac{3}{2})$ & Optimal & $\Theta(n)$ & $\Theta(n)$ \\ 
    min.~vertex coloring & $[\tfrac{3}{2},\infty)$ & Flexible & $\Theta(\log^{*} n)$ & $\Theta(\log^{*} n)$ \\
    \midrule
    \Cref{ex:max-dom-par}: & $[1,\tfrac{3}{2})$ & Optimal & $\Theta(n)$ & $\Theta(n)$ \\ 
    max.~domatic partition & $[\tfrac{3}{2},3)$ & Flexible & $\Theta(\log^{*} n)$ & $\Theta(\log^{*} n)$ \\ 
    & $[3,\infty)$ & Constant solution & $O(1)$ & $O(1)$ \\
    \midrule
    \Cref{ex:sloppy-col}: & $[1,2)$ & Optimal & $\Theta(n)$ & $\Theta(n)$ \\ 
    sloppy coloring  & $[2,3]$ & Flexible & $\Theta(\log^{*} n)$ & $\Theta(\log^{*} n)$ \\ 
    & $(3,100)$ & Constant fragment & $\Theta(\log^{*} n)$ & $O(1)$ \\ 
    & $[100,\infty)$ & Constant solution & $O(1)$ & $O(1)$ \\
    \bottomrule
  \end{tabular}
  \label{table:classification-ex}
\end{table}

\subsection{Future Directions and Open Questions}

Our results are easy to adapt to \emph{weaker} models: the results generalize essentially verbatim also to the \congest model \cite{peleg-2000-distributed-computing-a-locality-sensitive}, and they are easy to adapt to the port-numbering model, where we essentially lose the deterministic $\Theta(\log^*n)$ class. However, generalizing the results to \emph{stronger} classes leads to interesting open questions. One obvious candidate is the \qlocal model; however, whether symmetry-breaking problems require $\Omega(\log^* n)$ rounds in \qlocal is a major open question \cite{akbari-coiteux-roy-etal-2025-online-locality-meets}, and until this is resolved, we cannot fully characterize local optimization problems in the \qlocal model, either.

An orthogonal direction is generalizations of our questions to more diverse families of inputs. The first steps would be to generalize our results from directed cycles to undirected cycles, directed paths, and rooted trees. Already the case of undirected cycles leads to new technical challenges: constant-round deterministic algorithms for optimization problems in undirected cycles may achieve non-trivial approximations (consider e.g.\ the task of orienting the edges, and assign a higher cost to sink nodes). We believe that by introducing notions analogous to \emph{mirror-flexibility} from \cite{chang-studeny-suomela-2023-distributed-graph-problems} to our setting, it is also possible to eventually classify all these cases.

We have not considered inputs, and as soon as we have inputs, the classification will be \PSPACE-hard \cite{balliu-brandt-etal-2019-the-distributed-complexity-of}. Whether it is nevertheless decidable for local optimization problems with inputs is yet another open question for future work.

\section{Preliminaries}

We denote the set of positive integers by $\N^+$ and $\N^+ \cup \{ 0 \}$ by
$\N_0^+$.
Similarly, $\R_0^+$ is the set of non-negative reals.
We also write $[n] = \{ 1, \dots, n \}$ for the set of the first $n$ elements in
$\N^+$.

Moving on to graphs, we consider only cycles.
Without restriction, let us name the nodes of a directed cycle of length $n$
such that $V = {v_{0},v_{1},\ldots,v_{n-1}}$ and $(v_{i},v_{(i+1)\bmod n}) \in E$
for $i \in \{ 0, \dots, n-1 \}$. 
The \emph{neighborhood} of radius $r \in \N_0^+$ of a node is a function $N_{r}: V \to
V^{r+1}$ such that 
\[
  N_{r}(v_i) = (v_{i},v_{(i+1)\bmod n},\ldots,v_{(i+r)\bmod n}).
\]

\subsection{Locally Checkable Optimization Problems}
\label{sec:opt-lcl}

A \emph{locally checkable optimization problem} (\emph{opt LCL}) $\Pi$ on
directed cycles is a tuple $\Pi = (\Gamma, r, c, \aggr, \obj)$
where:
\begin{itemize}
  \item $\Gamma$ is a finite alphabet.
  \item $r \in \N_0^+$ is the (constant) \emph{radius}.
  \item $c: \Gamma^{r+1} \to \R_0^+ \cup \{\bot\}$ assigns every
  neighborhood a certain \emph{cost}.
  \item $\aggr \in \{\sum, \min, \max\}$ is an \emph{aggregation function}.
  \item $\obj \in \{\min, \max\}$ is an \emph{objective function}.
\end{itemize}
The label $\bot$ is used to indicate neighborhoods that are not valid. 
In the following, such neighborhoods and their cost can be immediately
disregarded. 
In order to distinguish different types of opt LCLs we use the shorthands like $\obj$-$\aggr$ (opt) LCL (e.g., min-sum LCL for an opt LCL where $\obj = \min$
and $\aggr = \sum$).

A \emph{solution} $s$ on a graph $G = (V,E)$ with at least one node is a
function $s: V \to \Gamma$. 
For the neighborhood $N_r(v_{i})$ of a node $v_{i}$, we define
\[
  s(N_r(v_{i})) = (s(v_{i}), s(v_{(i+1)\bmod n}), \ldots, s(v_{(i+r)\bmod n})).
\]
We let $S_{n}$ be the solution space containing all solutions for a problem on a
graph with $n$ nodes.
A solution $s$ to an opt LCL $\Pi = (\Gamma,r,c,\aggr,\obj)$ on a graph
$G=(V,E)$ is said to be \emph{valid} if and only if $c(s(N_r(v))) \neq \bot$
holds for every $v \in V$.
 
The \emph{optimal cost} of a problem $\Pi = (\Gamma,r,c,\aggr,\obj)$ with
$\abs{\Gamma}^{n}$ solutions $S_{n}$ on a directed cycle $G = (V,E)$ of length
$n$ is a function 
\[
  \opt(n) = \obj_{s \in S_{n}} \aggr_{v \in V} c(s(N_r(v))).
\]
Note that indeed $\opt$ solely depends on $n$ as for a given $n$ we know $G$.
A solution $s$ is said to be \emph{optimal} for some $n$ if and only if
\[
  \aggr_{v \in V} c(s(N_r(v))) = \opt(n).
\]

For any opt LCL $\Pi$ and a given $n$ with $\obj = \min$, a value $v$ with (at
least one) associated solution $S_v$ is an \emph{$\alpha$-approximation} if
the constant $\alpha \ge 1$ and $\alpha\cdot\opt(n) \geq v$. 
If $\obj = \max$, the condition is that $\frac{1}{\alpha}\opt(n) \leq v$. 
An algorithm is said to \emph{find an $\alpha$-approximation for $\Pi$} if it
finds an $\alpha$-approximation for any $n$.

\subsection{\boldmath The \local Model}
\label{sec:local-model}

We assume the usual distributed \local model of computation due to Linial
\cite{linial-1992-locality-in-distributed-graph-algorithms}.
The input graph $G$ is perceived as defining a distributed network where nodes
are computers and edges are communication links.
Each node is assigned a unique identifier from the set $[n^c]$, where $c > 1$
is some constant of our choosing.
Computation proceeds in rounds, where in each round a computer exchanges
messages of unbounded size with its neighbors and executes a procedure where it
decides whether to commit to an output or not and what to send on the next
round of communication.
We disregard time and space complexity of local computation and care only about
the number of communication rounds required to solve the problem at hand.
In this paper, our upper bounds give computable local procedures at every node
whereas for the lower bounds we may as well assume an arbitrary procedure.

Whenever we address the \randlcl model, we define an algorithm to produce an $\alpha$-approximation if the output is with high probability valid in terms of the problem and additionally approximates the optimal solution with the factor $\alpha$. Meaning, that only with probability $\leq \frac{1}{n}$ the result is a labeling that either is not valid or exceeds the allowed approximation.

\section{Defining Problem Parameters}\label{sec:problem-param}

In this section, we define the seven parameters $\bopt$, $\bflex$, $\bcoprime$, $\dflex$,
$\bgap$, $\dgap$, and $\bconst$ that fully capture the complexity of producing
approximations of an opt LCL $\Pi$.
Before we can do this in \cref{def:g-a,def:b-d-parameters,def:bcoprime}, we
first introduce some technical tools which we also require later in
\cref{sec:efficient}.

The main technical tool that we employ to analyze an opt LCL $\Pi$ is the
\emph{de Bruijn graph} associated with it.
The values of the problem parameters for $\Pi$ all come from characteristics of
certain walks in the de Bruijn graph associated with it.

\begin{definition}[De Bruijn Graph]
  Given an opt LCL $\Pi=(\Gamma, r, c', \aggr, \obj)$, its \emph{de Bruijn
  graph} is the digraph $G=(V,A)$ where:
  \begin{itemize}
    \item The nodes in $V$ are $(r+1)$-tuples $(s_1,\dots,s_{r+1}) \in
    \Gamma^{r+1}$.
    \item We draw an edge from $(s_1, \dots , s_{r+1})$ to $(t_1, \dots ,
    t_{r+1})$ if $t_i = s_{i+1}$ holds for every $i \in [r]$.
  \end{itemize}   
  Moreover we equip this graph with the cost function $c: V \to R(c')$, where $R(c')$ is the range of $c'$. The function $c$ attributes to a particular neighborhood $N \in \Gamma^{r+1}$ its cost $c'(N)$. We denote by $\gamma = |\Gamma|^{r+1}$ the number of nodes in $G$.
\end{definition}

\Cref{example_1_debruijn,example_2_debruijn,example_3_debruijn,example_4_debruijn,example_5_debruijn}
show the de Bruijn graphs associated with the problems from
\cref{ssec:intro-examples}.

\begin{figure}
  \begin{subfigure}{0.49\textwidth}
    \centering
    \includegraphics{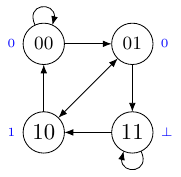}
    \caption{The de Bruijn digraph of \cref{ex:max-ind-set}.} 
    \label{example_1_debruijn}
  \end{subfigure}
  \hfill
  \begin{subfigure}{0.49\textwidth}
    \centering
    \includegraphics{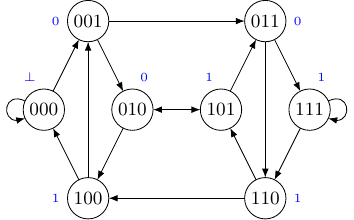}
    \caption{The de Bruijn digraph of \cref{ex:min-dom-set}.} 
    \label{example_2_debruijn}
  \end{subfigure}
  \vspace{1em}
  \begin{subfigure}{\textwidth}
    \centering
    \includegraphics{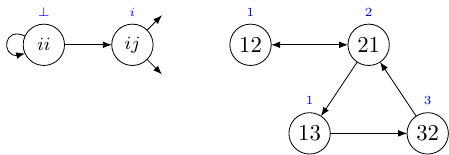}
    \caption{The de Bruijn digraph of \cref{ex:min-ver-col}. On the left-hand side the general structure of the graph. On the right-hand side the subgraph relevant to calculate the values in \cref{table:strategies-ex}.} 
    \label{example_3_debruijn}
  \end{subfigure}
  \vspace{1em}
  \begin{subfigure}{\textwidth}
    \centering
    \includegraphics[width=\textwidth]{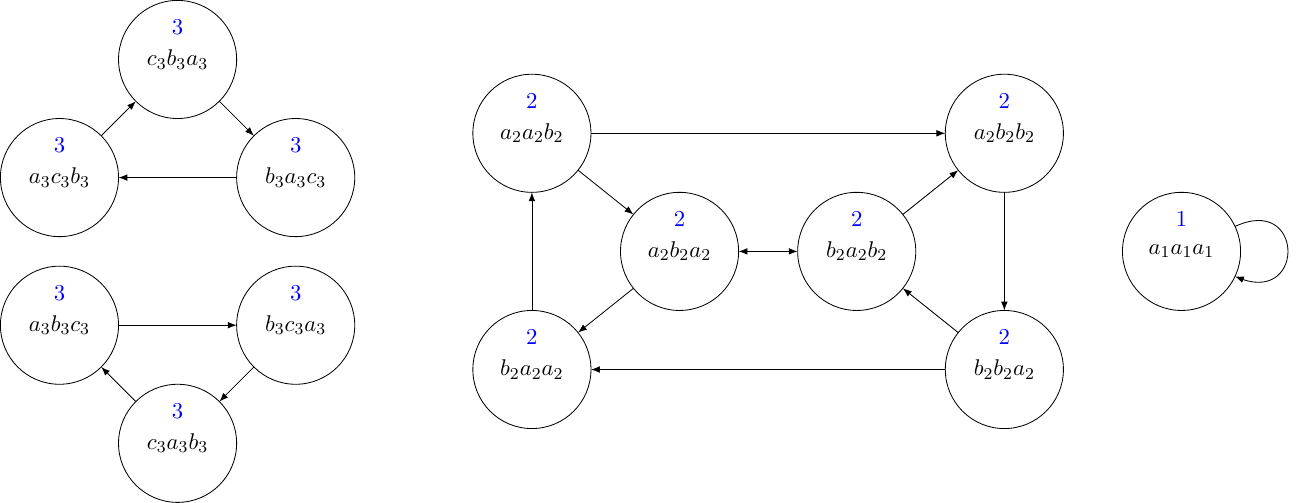}
    \caption{The de Bruijn digraph of \cref{ex:max-dom-par}.} 
    \label{example_4_debruijn}
  \end{subfigure}
  \vspace{1em}
  \begin{subfigure}{\textwidth}
    \centering
    \includegraphics{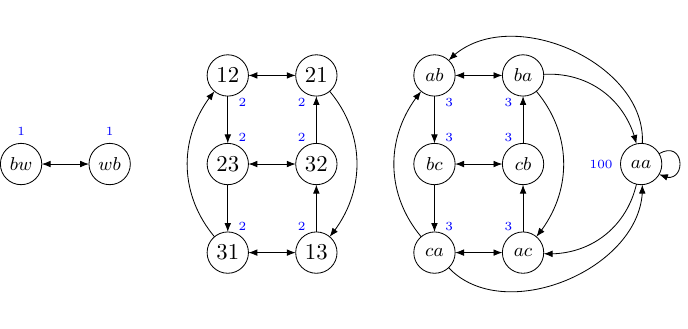}
    \caption{The de Bruijn digraph of \cref{ex:sloppy-col}.} 
    \label{example_5_debruijn}
  \end{subfigure}
  \caption{De Bruijn digraphs for
  \cref{ex:max-ind-set,ex:min-dom-set,ex:min-ver-col,ex:max-dom-par,ex:sloppy-col}.
  In blue, the cost of each neighborhood.
  Nodes with cost $\bot$ can be considered as non-existing in the rest of the paper.}
\end{figure}

Next we look at walks in the de Bruijn graph associated with $\Pi$ as well as
the total cost incurred by following such walks.

\begin{definition}[$(u,v)$-walk]
  A sequence of nodes $W = w_0, \dots, w_k$ with $\forall i, \, w_i \in V$ is a
  \emph{$(u,v)$-walk} if $w_0 = u$, $v=w_k$ and $\forall i, (w_i,w_{i+1})$ is an
  arc of $G$. When $u=v$, we refer to $W$ as a closed walk, or, in particular,
  as a closed $v$ walk.
\end{definition}

\begin{definition}[Cost of a walk]
  Let $W = w_0 \dots w_k$ be a walk with $k \ge 1$. If there is some $i \ge 1$ for which
  $c(w_i)=\bot$, then the \emph{cost} of $W$ is $c(W)=\bot$. 
  Otherwise, if $\aggr = \sum$, $c(W)=\frac{1}{k}\sum_{i=1}^k c(w_i)$ and $c(W)=\aggr_{i \in [1, \dots, k] } c(w_i)$ else. 
\end{definition}

Note that we ignore the value of $c(w_0)$ systematically to avoid counting some element twice when $W$ is a closed walk.

Finally, we define \emph{flexible nodes} and \emph{components}, which are relevant for flexible strategies.

\begin{definition}[Flexible nodes and components]\label{def:flex-node}
  A node $v$ of $G$ is \emph{flexible} if there exists $K \in \mathbb{N}$ such
  that $\forall k \ge K$, there exists a closed $v$ walk of length exactly $k$.
  We say in particular that $v$ is \emph{$K$-flexible}. 

  A strongly connected component $S \subset{V}$ of $G$ is a \emph{flexible
  component} if there is a node $v \in S$ such that $v$ is flexible. 
\end{definition}

We are now ready to define the seven parameters $\bopt$, $\bflex$, $\bcoprime$, $\dflex$,
$\bgap$, $\dgap$, and $\bconst$ that capture the key properties of an opt LCL
$\Pi$.
Recall $\opt(n)$ denotes the value of an optimal solution for a cycle of $n$
nodes.
Basically, we relate each of the $\beta_a$ values to the optimal solution in a
certain subclass of the solution space.
These subclasses are implicitly defined by taking appropriate subgraphs of the
de Bruijn graph associated with $\Pi$.
Meanwhile, $\dflex$ and $\dgap$ are boolean values that depend on certain
specific properties of $\Pi$.

\begin{definition}[$\Gopt$, $\Gflex$, $\Ggap$, $\Gconst$]
  \label{def:g-a}
  Let $G$ be the de Bruijn graph associated with an opt LCL.
  We define a set of (possibly empty) subgraphs $\mathcal{G} = \{\Gopt, \Gflex,
  \Ggap, \Gconst\}$ for $\Pi$, containing:
  \begin{enumerate}
    \item $\Gopt = G$.
    \item $\Gflex \subseteq \Gopt$ containing only the flexible components of $\Gopt$.
    \item $\Ggap \subseteq \Gflex$ containing only components of $\Gflex$ that contain a self-loop.
    \item $\Gconst \subseteq \Ggap$ containing only nodes of $\Ggap$ that have a self-loop.
  \end{enumerate}
\end{definition} 

\begin{definition}[$\bopt$, $\bflex$, $\bgap$, $\bconst$, $\dflex$, $\dgap$]
  \label{def:b-d-parameters}
  Let $\Pi=(\Gamma, r, c', \aggr, \obj)$ be an opt LCL.
  For each $G_a \in \mathcal{G}$, we define $\mathcal{C}(G_a)$ as the set of
  cycles in $G_a$ and the real-valued parameter 
  \[
    \beta_a = \obj_{\substack{C \in \mathcal{C}(G_a) \\ c(C) \neq \bot}} c(C). 
  \]
  Note this includes self-loops, which are cycles of length one.
  If $\mathcal{C}(G_a)$ is empty, then $\beta_a = \bot$.

  Meanwhile, the parameters $\dflex$ and $\dgap$ are Booleans:
  \begin{enumerate}
    \item $\dflex$ is false if there are two closed walks $W_1, W_2$ in $\Gflex$ of
    coprime length, sharing a node, so that $c(W_1) = c(W_2) = \bflex$; otherwise, we set
    $\dflex$ to true. 
    \item $\dgap$ is false if $\bgap = \bconst$, and true otherwise.
  \end{enumerate}
\end{definition}

As an example, for $\min\sum$ problems we can connect $\beta_a$ with concrete
solutions to $\Pi$ as follows:
\begin{enumerate}
  \item $\bopt$ is the best possible relative value that can be achieved
  infinitely often (i.e., asymptotically) by the optimum over all cycle lengths.
  \item $\bflex$ is the (again, asymptotically) worst such value.
  \item $\bgap$ is the best possible value for solutions that admit constant
  neighborhoods, that is, where a single, fixed label appears.
  \item Finally, $\bconst$ is the best possible (relative) value for constant
  solutions, that is, where the same fixed label is used \emph{everywhere}.
\end{enumerate}

Note that, for example when summing, if $k$ is the length of the cycle providing
$\bopt$, then for any $n$ we have $n \bmod k = 0$ implies $\frac{\opt(k)}{k} =
\frac{\opt(n)}{n} = \bopt$.
In addition, for large $n$, we have $\bopt\cdot n \leq \opt(n) \leq \bflex n +
O(1)$. It is possible that $\dflex$ is false and $\bflex = \bopt$.

\begin{definition}[$\bcoprime$]
  \label{def:bcoprime}
  To define $\bcoprime$, we let
  \[
    V_\le(r) = \left\{ v \in V \mid c(v) \neq \bot \land c(v) \le r \right\}
  \]
  and
  \[
    V_\ge(r) = \left\{ v \in V \mid c(v) \neq \bot \land c(v) \ge r \right\}
  \]
  and let $G_\le(r)$ and $G_\ge(r)$ be the subgraphs of the de Bruijn graph $G$
  induced by $V_\le(r)$ and $V_\ge(r)$, respectively.
  We then define $\bcoprime$ for min-max and max-min problems as follows:
  \begin{itemize}
    \item For min-max problems, $\bcoprime$ is the smallest $r \in \R_0^+$ for
    which the graph $G_{\le}(r)$ contains two walks of coprime lengths that
    have a node in common.
    \item For max-min problems, $\bcoprime$ is the largest $r \in \R_0^+$ for
    which the graph $G_{\ge}(r)$ contains two walks of coprime lengths that
    have a node in common.
  \end{itemize}
\end{definition}

\section{Problem Parameters Are Computable Efficiently}\label{sec:efficient}

In this section, denote by $G$ the de Bruijn digraph of the given opt LCL. Recall that $\gamma$ is its number of nodes. We will assume that $G$ is pruned of its vertices whose value is $\bot$. In practice, it may mean that $G$ may have multiple strongly connected components, no self-loop, etc. Next we present the necessary result to prove that all seven problem parameters can be computed efficiently. Note that $\bgap$, $\dflex$ and $\dgap$ are only relevant to sum opt LCLs, so we do not need to compute them otherwise.

Let $\mathcal{C}_{\operatorname{opt}}$ be the set of closed walks that are optimal regarding
the objective function. In the event that $\mathcal{C}_{\operatorname{opt}} = \emptyset$, it is
impossible to form a closed walk in $G$, and hence it is impossible to solve
$\Pi$. Hence we may assume $\mathcal{C}_{\operatorname{opt}} \ne \emptyset$.

\begin{lemma}\label{lem:wopt-is-poly}
  If $\mathcal{C}_{\operatorname{opt}} \ne \emptyset$, there exists a closed walk in $\mathcal{C}_{\operatorname{opt}}$ of length at most $\gamma$.
\end{lemma}

\begin{proof} 
Let $w \in \mathcal{C}_{\operatorname{opt}}$ be a closed $v$ walk of length at least $\gamma+1$. We prove that we can find a walk $w' \in \mathcal{C}_{\operatorname{opt}}$ such that $|w'|<|w|$. We repeat this argument on $w'$ until we eventually find a walk of length at most $\gamma$.

As $|w|\ge \gamma+1$, by the pigeonhole principle, either $v$ is visited three times or there exist some $v^* \ne v$ visited twice.
In the first case, we split $w$ in two closed $v$ walks at $v$, $L_1$ and $L_2$.
Observe that both $L_1$ and $L_2$ have their length in $[1, |w|-1]$, and both
have their cost equal to $c(w)$, as otherwise $w \notin \mathcal{C}_{\operatorname{opt}}$:
If $c(L_1) < c(w)$ (resp., $c(L_1) > c(w)$ in the case of maximization), $w$ is
not optimal anymore. 
Moreover, if $c(L_1) > c(w)$ (resp., $c(L_1) < c(w)$ for maximization), then
either $c(w) > c(w)$ (resp., $c(w) < c(w)$) for $\aggr = \sum$ or $c(L_2)
< c(w)$ (resp., $c(L_2) > c(w)$), which is again a contradiction.
Hence, both of them belong to $\mathcal{C}_{\operatorname{opt}}$.   

If there is $v^* \ne v$ visited twice, the proof is quite similar. Split $w$ into three walks at $v^*$: $L_1 = v, \dots, v^*$, $L_2 = v^*, \dots, v^*$ and $L_3 = v^*, \dots, v$.
Now consider the two closed walks $L_1+L_3$ and $L_2$. Their lengths are in $[1, |w|-1]$ and they both must have the same cost as $w$; hence they belong to $\mathcal{C}_{\operatorname{opt}}$.
\end{proof}

The definition of $\mathcal{C}_{\operatorname{opt}}$ suggests that the optimal elements of the de Bruijn graph are closed walks and not only simple cycles. In the next lemma we prove that we can actually think only about cycles:

\begin{lemma}\label{lem:cycle-vs-walk}
If there is a closed walk $W$ of length at most $\gamma$ with cost $\beta$ in the de Bruijn graph then there is a cycle of length at most $\gamma$ with cost at least as good as $\beta$ in terms of the optimization function. 
\end{lemma}

\begin{proof}
If $W$ is a cycle, the proof is trivial. Else, divide $W$ into any two subwalks $W_1$ and $W_2$. If $W_1$ is better than $W_2$ regarding the objective function, proceed this process iteratively with $W_1$, else with $W_2$, until the closed walk remaining is a cycle. It is apparent that $W$ cannot yield a better result on any optimization function than the better of the two of its subwalks, no matter if we look at an average cost or a min / max cost. 
\end{proof}

In order to compute $\bflex$ and $\bcoprime$ it is crucial to identify flexible nodes of the graph. To figure out if a given node is flexible we prove the following lemma:

\begin{lemma}\label{lem:coprime-cycles-equals-flexible-components}
   Node $v$ is flexible if and only if $v$ belongs to two closed walks of coprime lengths. 
\end{lemma}

\begin{proof}
  Let $A$ and $B$ be two closed $v$ walks with $a = |A|$ and $b = |B|$ being coprime. In that case $K = (a-1)(b-1)$ is the Frobenius number, meaning that all numbers greater than $K$ can be obtained by a positive linear combination of $a$ and $b$, hence $v$ is a $K$-flexible node. 
  
  Now let $v$ be a flexible node. Hence there exist $K>0$ such that for all $k\ge K$, there exists a closed $v$ walk of length exactly $k$. Pick any $a \ge K$ and $b \ge K$ such that $a$ and $b$ are coprime. 
\end{proof}

Figuring the value of $\bflex$ is in principle easy: for each connected component, we can compute the cost of the optimum cycle in that component, and test which components are flexible. Testing the flexibility of a component can be achieved by considering walks of length $O(\gamma)$ \cite{chang-studeny-suomela-2023-distributed-graph-problems}. We then set $\bflex$ accordingly.

The hard part of our work is to decide whether $\dflex$ is false or not. $\dflex$ would be false only if there are two closed walks of coprime length, sharing a vertex and of cost $\bflex$. We adapt the procedure to test flexibility by \cite{chang-studeny-suomela-2023-distributed-graph-problems} to take into account the cost function of the de Bruijn graph in the next two lemmas. 

\begin{definition}
  If $\Pi$ is an $\obj$-$\sum$ opt LCL, then $G_{\bflex} = (V_{\bflex},A_{\bflex})$, where
  \begin{align*}
    V_{\bflex} &:= \left\{\ v \in V \mid v \text{ belongs to a closed walk of cost } \bflex\right\}, \\
    A_{\bflex} &:= \{\ (u,v) \in A \mid \text{$(u,v)$ belongs to a closed walk of cost $\bflex$}\}.
  \end{align*}
\end{definition}

It would be reasonable to think that having very long walks might be good in order to optimize the cost of a solution. We prove in our next lemma that this is not the case: all the walks of $G_\bflex$ have cost exactly $\bflex$, and hence the length of the walk does not matter. This lemma, coupled with the fact that testing flexibility can be done by considering walks of limited length will allow us to efficiently test whether $\dflex$ is false or not.
As $\dflex$ is a parameter defined only for $\obj$-$\sum$ opt LCLs, there is no analogous definition of $G_\bflex$ for $\min$-$\max$ and $\max$-$\min$ opt LCLs. 

\begin{lemma}\label{lem:zero_pot-minsum}
  Let $\Pi$ be an $\obj$-$\sum$ opt LCL. Any closed walk in the graph $G_\bflex$ has cost exactly $\bflex$.
\end{lemma}

\begin{proof} 
  Now let us assume that $\Pi$ is an $\obj$-$\sum$ opt LCL. Call the shifted cost function $c_\bflex\colon V_\bflex \to
  \mathbb{R}^+$ defined as $c_\bflex(v) = \frac{c(v)}{\bflex}-1$. Call the
  cost of a walk $W$ computed using $c_\bflex$ instead of $c$ the \emph{shifted
  cost} of $W$. It is clear that a walk of cost $\bflex$ has shifted cost zero
  (as $k\bflex = \sum_1^k c(a)$ if and only if $0 = \sum_1^k( c(a)/\bflex -1)$ for $\obj$-$\sum$-opt LCLs).
  Moreover, by definition of $\bflex$, all closed walks must have shifted cost at least $0$ if $\obj = \min$ (resp. at most $0$ if $\obj = \max$).
  For simplicity, we can imagine that $G_\bflex$ is constructed by adding one by
  one a finite number of simple closed walks $W_1, W_2, \dots$ of shifted cost
  zero.
  (Note that we do not actually need to design an algorithm to build this
  sequence.)

  We now prove the lemma by induction. Let $W$ be a closed walk of $G_\bflex$. Assume that $W$ exists solely on the graph spanned by $W_1$. It is clear that $W_1$ being a simple closed walk, $W$ consists of one or multiple instances of $W_1$ and hence $c_\bflex(W) = 0$.

  Suppose now that for some $k \in \mathbb{N}$, any closed walk on the graph spanned by $\cup_i^k W_i$ has
  shifted cost exactly zero. We want to prove that any
  closed walk $W$ on $\cup_i^{k+1} W_i$ has shifted cost exactly $0$. We use the
  following notation to represent $W$:  \[a_1^1, \dots, a_{k_1}^1 = b_1^1, \dots,
  b_{\ell_1}^1 = a_1^2, \dots, a_{k_2}^1, \dots,\]  where $a_{k_i}^i = b_1^i$,
  $b_{\ell_i}^i = a_1^{i+1}$ and $b_i^j$ are nodes of $W_{k+1}$, and $a_i^j$ are
  other nodes. The reader can refer to \cref{fig:flexgraph} as a visual aid for the rest of the proof.
  
  Consider a walk $W'$ defined in the following way: it is built following $W$ but each time we reach a frontier node $a_{k_i}^i = b_1^i$, we add a first segment going from $a_{k_i}^i$ to $a_1^{i+1}$ through nodes in $\cup_{j=1}^{k} W_j$ then a second segment going from $a_1^{i+1} = b_{\ell_i}^i$ to $b_1^i$ through nodes of $W_{k+1}$. Then, $W'$ resumes its visit of $W$. Note that it is possible to reorient the way in which $W'$ is travelled so that it can be seen as a visit of $\cup_i^{k+1} W_i$, then when it reaches a frontier node $b_1^i$, it may go through $W_{k+1}$ entirely, then resumes its visit of $\cup_i^{k} W_i$. Then it is clear that $c_\bflex(W') = 0$ as it consists of multiple travels through $W_{k+1}$ and a walk on $\cup_i^k W_i$. Also, we have $W \subseteq W'$. If we had $c_\bflex(W) > 0$ (resp., $c_\bflex(W) < 0$), then in some way $W'$ would be compensating in some part of $W'-W$. However, $W'-W$ consists of a collection of cycles of the form $b_{\ell_i}^i \dots b_1^i$ plus some segments going from $a_{k_i}^i$ to $a_1^{i+1}$ in $\cup_i^k W_{i}$. If $c(W)>0$ (resp., $c(W) < 0$), then one of those cycles would be negative (resp., positive) shifted cost, which is a contradiction with the definition of $\bflex$. Hence it follows that all closed walks on $\cup_i^{k+1} W_i$ have cost exactly $\bflex$.
\end{proof}
\begin{figure}
    \centering
    \begin{subfigure}{.5\textwidth}
  \centering
  \includegraphics[width=\textwidth]{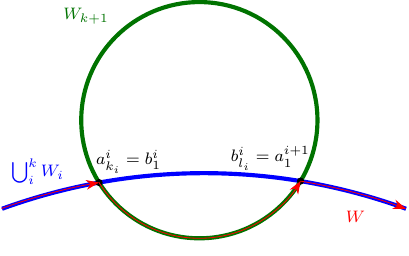}
\end{subfigure}
\begin{subfigure}{.5\textwidth}
  \centering
  \includegraphics[width=\textwidth]{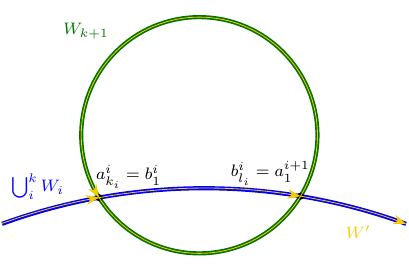}
\end{subfigure}
    \caption{Illustration of the construction of $W'$ using $W$ in \cref{lem:zero_pot-minsum}. Each time $W$ uses some part of $W_{k+1}$, we add at the end of it the unused part of $W_{k+1}$ (the green upper part) and any path going from $a^i_{k_i}$ to $a_1^{i+1}$ in $\cup_i^k W_{i}$. This way it is clear that $W'$ consists of possibly multiple instances of $W_{k+1}$ and some closed walk on $\cup_i^k W_{i}$, making the shifted cost of $W'$ necessarily $0$.}
    \label{fig:flexgraph}
  \end{figure}

As mentioned, for testing if a digraph has flexible nodes it is sufficient to
look at closed walks of lengths at most $2\gamma+1$ as shown in
\cite{chang-studeny-suomela-2023-distributed-graph-problems}. We provide the
rewriting of the proof in \cref{apdix:flex-comprime}. We can directly use this result to prove that in $G_\le (\bcoprime)$ (resp. $G_\ge(\bcoprime)$) there are two closed walks of coprime lengths sharing a vertex of length at most $2\gamma+1$ and cost at most $\bcoprime$ (resp. at least) as all walks of $G_\le(\bcoprime)$ (resp. $G_\ge(\bcoprime)$) must have cost at most (resp. at least) $\bcoprime$.
Thanks to \cref{lem:zero_pot-minsum}, we know that all the walks of $G_\bflex$ have cost exactly $\bflex$, so we can also apply the same result. Then we can test if there are two closed walks of coprime lengths sharing a vertex and having cost $\bflex$ by only looking at walks of length at most $2\gamma+1$; hence deducing efficiently whether $\dflex$ is false or not.

\begin{lemma} \label{cor:wflex-is-poly}
Let $G$ be the de Bruijn graph of an opt-sum LCL. If there are two walks of coprime length sharing a vertex and having cost $\bflex$ in $G$, then there are two such walks of length at most $2\gamma +1$.
\end{lemma} 

\begin{proof}
  Apply \cref{lem:2gamma-walks-flexibility} to $G_{\bflex}$ as it is defined in \cref{lem:zero_pot-minsum} and note that those walks exist in $G_{\bflex}$ if and only if they exist in $G$.
\end{proof}

\begin{lemma}\label{lem:minsum-maxsum-are-poly}
  For any opt LCL, if they exist, we can compute the values of $\bopt$, $\bflex$, $\dflex$, $\bcoprime$, $\bgap$, $\dgap$, and $\bconst$ in time polynomial in $|\Gamma|$.
\end{lemma}

\begin{proof}
  First observe that for any LCL we have $r = O(1)$ and hence any algorithm running in polynomial in $|\Gamma|$ runs in polynomial in $\gamma = |\Gamma|^{r+1}$.

  Let's start by addressing the more complex case of sum-opt LCLs:

  We will use the following subroutine:  
  For $s \in G$, let $A(s,K)$ be the following process: denote if it exists, by
  $F(s,v,k)$ an optimal $(v,s)$ walk of length exactly $k$. Initialize as
  non-existing $F(s, v,0)$ for $v \ne s$ and $F(s, s, 0) = \{s\}$. For all $v$
  and $k\ge 0$ we can compute $F(s,v,k+1)$ using the formula 
  \[
    \arg \obj_{u \in V\cap N^-(v)} \big( c(F(s,u,k)+u) \big), 
  \]
  where $N^-(v)$ represents the predecessors of $v$. Repeat this process until $k = K$. Then output the list of tuples $(k, F(s,s,k))$ for each value of $k$ for which $F(s,s,k)$ exists. $A(s,K)$ computes all the closed walks containing $s$ of length $k$, for all $k \le K$  that are optimal regarding the objective function. It can return multiple elements (at most $K$ tuples). 

  When $K$ is polynomial in $\gamma$, $A(s,K)$ is clearly also polynomial in $\gamma$. We can then use it to obtain all the desired values:
  \begin{itemize}
    \item $\bopt$: run $A(s,\gamma)$ for all $s$ then figure for which $s,k$ we have an optimal $c(F(s,s,k))$. Correctness is ensured by \cref{lem:wopt-is-poly}. Naturally, this procedure will detect also cycles among the closed walks. \cref{lem:cycle-vs-walk} ensures that the value we find will belong to at least one cycle, meeting the definition of $\bopt$. The procedure takes polynomial time in $\gamma$ as we call $A(s,\gamma)$ only $\gamma$ times.
    \item $\bflex$ and $\dflex$: first, determine the set of strongly connected components of $G$. For each of them, determining flexibility can be done by observing the output of $A(s,2\gamma+1)$ for all $s$; as if $s$ is flexible it must belong to two closed walks of coprime lengths (\cref{lem:coprime-cycles-equals-flexible-components}) and there must be two such walks of lengths at most $2\gamma+1$ (\cref{lem:2gamma-walks-flexibility}). Now, for nodes in flexible components, we can determine the value of $\bflex$ by running $A(s,\gamma)$ if $s$ is in a flexible component and pick the component that has the best cycle. As for $\dflex$, it suffices to look at the output of $A(s,2\gamma+1)$ in the good component and check if there are two walks containing $s$, of coprime lengths and cost $\bflex$. In short, all we have to do is to run $A(s,2\gamma+1)$ once for all nodes and all the information we need is in that output. Calling $A(s,2\gamma+1)$ for all $s$ is clearly polynomial in $\gamma$ and all the further analysis of the output is also polynomial.
    
    \item $\bgap$ and $\dgap$: Identify flexible components like explained above, but limit the considered components further by requiring them to contain one node with a self-loop. This procedure is clearly polynomial in $\gamma$ and obtains $\bgap$. Let $S_{\operatorname{loop}}$ be the set of all nodes with a self-loop. Then, $\dgap$ is false if and only if $S_{\operatorname{loop}}$ contains a node with a self-loop of cost $\bgap$. 
    \item $\bconst$: run $A(s,1)$ for every $s$. If $A(s,1)$ is non-empty, take the walk of optimal value (depending on $\obj$).
  \end{itemize}

  As for min/max opt LCLs, the procedure is much easier. Start by listing in $L$ all costs $c(s)$ for all nodes and sort $L$ in ascending if $\obj = \min$ (resp. descending if $\obj = \max$) order. Then, for all $\ell \in L$, we can construct the graphs $G_\le(\ell)$ and $G_\ge(\ell)$ in time polynomial in $\gamma$. Now, do the following:
  For all $\ell \in L$ in ascending (resp. descending order), do the same procedure as for sum-opt LCLs to determine $\bopt$. If the procedure fails (meaning we could not find a closed walk of length at most $\gamma$), keep going progressing in $L$. The first time the procedure succeeds, we have found $\bopt=\ell$. When $\bopt$ is found we can do a similar procedure to find $\bcoprime$: build $G_\le(\ell)$ (resp. $G_\ge(\ell)$) and run $A(s,2\gamma+1)$ for all $s$. As for $\obj$-$\sum$ opt LCLs, this fully determines which components are flexible. The first time a flexible component is found, we have $\bcoprime = \ell$. We repeat the same idea to find $\bconst$. Regarding the complexity, we can notice that all the procedure that were polynomial in $\gamma$ for sum-opt LCLs are called at most $|L| \le \gamma$ times, hence the whole algorithm runs in time polynomial in $\gamma$.
\end{proof}

By showing we can secure the problem parameters efficiently, we can arrive to our main theorem:

\begin{theorem}\label{theo:complex_by_beta}
Given the values of $\bopt$, $\bflex$, $\dflex$, $\bcoprime$, $\bgap$, $\dgap$, and $\bconst$, the lower and upper bound for the complexity of an algorithm solving $\Pi$ or any approximation of $\Pi$ can be derived directly, so classifying $\Pi$ can be done in polynomial time.
\end{theorem}

While with \cref{lem:minsum-maxsum-are-poly} we have proven already one part of this theorem, the whole next section is dedicated to prove the remaining part.

\section{Problem Parameters Determine Distributed Complexity}\label{sec:alg}

In this final section, we show that the complexity of any given problem for
any approximation ratio $\alpha$ can be determined using the
parameters from \cref{sec:problem-param}.
We do this by showing a series of lower and upper bounds in both deterministic
and \randlcl based on these parameters.
The result of all these observations is that every opt LCL falls into one of
five different classes based on its optimal complexity.
An overview of these classes, which we label from A down to E, is provided in
\cref{table:complexclasses}.

Throughout the previous sections we had to take into account that minimization
and maximization problems behave in distinct ways, for example when referring to
costs. 
However, the results we describe in this section can be readily ported from one
case to the other by simple methods such as flipping signs or inverting
fractions in the relevant inequalities. 
For example, an inequality $\alpha\bopt > \bflex$ in the minimization case
becomes $\alpha / \bopt > 1/\bflex$ in the maximization one (since we are
interested in $1/\alpha$ as a factor instead of $\alpha$).
Then we can simply recast the quantities $1/\bopt$ and $1/\bflex$, respectively,
as, say, $\beta'_{\operatorname{opt}}$ and $\beta'_{\operatorname{flex}}$ and
proceed with the same argument.
For this reason, in the text below we handle only the case of minimization
problems and skip any mention of maximization problems altogether.

\begin{table}
  \centering
  \caption{Any $\alpha$-approximation of an opt LCL $\Pi$ falls into one of these
  five complexity classes, depending on the structure of $\Pi$ and the value of
  $\alpha$.}
  \begin{center}
    \begin{tabular}{lll}
      \toprule
      & \multicolumn{2}{c}{Complexity} \\
      \cmidrule(r){2-3}
      Class & \detlcl & \randlcl \\
      \midrule
      A & $O(1)$ & $O(1)$ \\ 
      B & $\Theta(\log^{*} n)$ & $O(1)$ \\
      C & $\Theta(\log^{*} n)$ & $\Theta(\log^{*} n)$ \\
      D & $\Theta(n)$ & $\Theta(n)$ \\
      E & unsolvable & unsolvable \\
      \bottomrule
    \end{tabular}
  \end{center}
  \label{table:complexclasses}
\end{table}

\subsection{Lower Bounds}

First we prove lower bound results that show that \cref{table:complexclasses} is
complete, that is, that there is no opt LCL whose optimal complexity lies
in-between any of the classes A--E.

\begin{lemma}\label{lem:flexexpect}
  Assume there is a \randlcl algorithm that finds an $\alpha$-approximation of some sum-opt LCL $\Pi$ in $o(n)$ rounds.
  Then $\alpha\bopt \ge \bflex$ holds, where the inequality is strict if and
  only if $\dflex$ is true.
\end{lemma}

Since \randlcl is stronger than \detlcl, this directly implies a lower bound of
$\Omega(n)$ for \detlcl as well.

\begin{proof}
  It will be convenient to make the statement of the claim stronger so that it holds even if the algorithm knows a factor-$2$ approximation of $n$, the number of nodes; that is, all nodes get to see the same value $N$ such that $N/2 \le n \le N$. As this lemma is a lower bound, making the algorithm more powerful will only make the result stronger.

  With knowledge of $n$, we can then switch from the \randlcl model (where we have unique identifiers) to the randomized PN model (port-numbering model, where nodes are anonymous): if there is a \randlcl algorithm, we can simulate it in the randomized port-numbering model if each node first uses randomness to generate a fake unique identifier uniformly at random from $\{1,\dotsc,n^c\}$ for a suitable constant $c$; these fake identifiers will be globally unique w.h.p.\ (by union bound), so if the original algorithm is a Monte Carlo algorithm that works w.h.p.\ assuming unique identifiers, the new algorithm is also a Monte Carlo algorithm that works w.h.p.\ assuming unique identifiers.

  So now we can assume that $A$ is a randomized PN-model algorithm that w.h.p.\ outputs an $\alpha$-approximation of $\Pi$. Furthermore, we can assume that our port numbering is derived from the orientation of the cycle, so all $o(n)$-radius neighborhoods are isomorphic.

  Fix some estimate $N$, some input cycle $C$, some node $v$, and consider the following experiment. Apply $A$ to $C$, condition on the high-probability event that $A$ does not fail, and look at the output of nodes $N_{r+1}(v)$ in the local neighborhood of $v$; let $X_0, \dotsc, X_r$ be these outputs (which are now random variables). Then calculate the expected value $\eta = \expectation[c(X_0, \dotsc, X_r)]$ of the cost function. We make two observations:
  \begin{enumerate}
    \item As all $o(n)$-radius neighborhoods in all input cycles are isomorphic, the output distribution of $X_0, \dotsc, X_r$ cannot depend on $C$ or $v$, but it only depends on $N$; hence in what follows we will denote this value with $\eta(N)$ to emphasize that it is a function of $N$.
    \item If $C$ is an $n$-cycle, then $\eta(N) n \ge \opt(n)$; otherwise algorithm $A$ (when it succeeds) would produce an output that is in expectation cheaper than the optimal solution.
  \end{enumerate}
 
  Now let $N/2 \le n_1 < n_2 \le N$ be sufficiently large values such that:
  \begin{itemize}
    \item in $n_1$-cycles we have $\opt(n_1) = \bopt n_1$,
    \item in $n_2$-cycles the optimum is $\opt(n_2) = \bflex n_2 + O(1)$.
  \end{itemize}
  We will use $N$ as the factor-$2$ approximation of the number of nodes $n$ that we reveal to the algorithm, so the $o(n)$-round algorithm does not know if we are in an $n_1$-cycle or $n_2$-cycle.
  
  Now in $n_2$-cycles the algorithm (conditioned on not failing) will find a solution with the expected total cost $\eta(N) n_2$, which has to be at least $\bflex n_2 + O(1)$, as better solutions do not exist.

  On the other hand, in $n_1$-cycles the algorithm (conditioned on not failing) will find a solution with the expected total cost $\eta(N) n_1$. As the algorithm is then supposed to find an $\alpha$-approximation, the cost of any solution produced by the algorithm in $n_1$-cycles has to have cost at most $\alpha \opt(n_1) = \alpha \bopt n_1$, and in particular the expected cost cannot be higher than that.

  Put together, we have
  \[
    \bflex + \frac{O(1)}{n_2} \le \eta(N) \le \alpha \bopt,
  \]
  from which the claim follows.
\end{proof}

\begin{lemma}\label{lem:lowercoprime}
If there is a \randlcl algorithm that finds an $\alpha$-approximation of some $\min$-$\max$ LCL $\Pi$ in $o(n)$ rounds then $\alpha\bopt \ge \bcoprime$.
\end{lemma}

\begin{proof}
Following the proof of \cref{lem:flexexpect}, replace the expected output cost with the largest output cost that happens with a non-negligible probability; that is, instead of $\eta(N)$, consider $\eta'(N)$, which is the largest value that the algorithm outputs in those cases in which it succeeds. Additionally replace $\bflex$ with $\bcoprime$ throughout the proof, for example change ``in $n_2$-cycles the optimum is $\opt(n_2) = \bflex n_2 + O(1)$'' to ``in $n_2$-cycles the optimum is $\opt(n_2) = \bcoprime$'' as we are not summing.

An alternative approach can be found in \cite{chang-studeny-suomela-2023-distributed-graph-problems} (Lemma 6.3.), which can be applied also to the \randlcl model.
\end{proof}

\begin{lemma}\label{lem:symmbreak}
  If an opt LCL $\Pi$ can be $\alpha$-approximated in $o(\log^\ast n)$ rounds
  in the \randlcl model, then $\alpha\bopt \geq \bgap$.
  Additionally, if $\dgap$ is true, then this inequality is strict.
\end{lemma}

The proof follows more or less directly from well-known results on symmetry
breaking and, in particular, the classical $\Omega(\log^\ast n)$ lower bound for
$3$-coloring \cite{naor-1991-a-lower-bound-on-probabilistic-algorithms-for}.
Note that this lower bound holds even when we are allowed to err, for instance,
with probability $1/\Omega(\log n)$ (instead of $1/n$ as usual).

\begin{proof}
  We prove the contrapositive, that is, assuming $\alpha \bopt < \bgap$ we have
  that any \randlcl algorithm $\mathcal{A}$ that $\alpha$-approximates $\Pi =
  (\Gamma, r, c', \aggr, \obj)$ must have locality $\Omega(\log^* n)$ rounds.
  Indeed, suppose we have such an algorithm $\mathcal{A}$.
  In light of the preceding \cref{lem:flexexpect}, we may additionally assume
  that $\alpha \bopt \ge \bflex$ holds.
  Put together with the previous inequality, this means that with high
  probability the output of $\mathcal{A}$ is contained in $\Gflex$ but not in
  $\Ggap$ for infinitely many values of $n$.
  In particular, this means the solutions produced by $\mathcal{A}$ avoid the
  components of $\Gflex$ that contain a self-loop, which can be translated to
  avoiding constant substrings of length $r+1$ (i.e., the strings of
  $\Gamma^{r+1}$) with high probability.
  Therefore, in the output of $\mathcal{A}$ there are at least two distinct
  labels in each neighborhood of size $r$ of every node (again, with high
  probability).  
  From this it is possible to extract a ruling set and, by known techniques
  (see, e.g., \cite{awerbuch-goldberg-etal-1989-network-decomposition-and}), a
  $3$-coloring of the cycle.
  Since these reduction steps to $3$-coloring can be performed with $O(1)$
  locality and the compounded error remains smaller than $1/n^{\Omega(1)}$, from
  the aforementioned lower bound
  \cite{naor-1991-a-lower-bound-on-probabilistic-algorithms-for} it follows that
  $\mathcal{A}$ has locality $\Omega(\log^* n)$.

  To prove the addendum concerning $\dgap$, suppose that we have $\alpha \bopt =
  \bgap < \bconst$.
  By the material in \cref{sec:efficient}, this inequality implies that
  self-loops are strictly more expensive than the minimal walks in $\Ggap$, and
  hence a minimal walk in $\Ggap$ must avoid self-loops entirely.
  If $\mathcal{A}$ produces such solutions, then it must again avoid constant
  substrings of length $r+1$ altogether, and thus the rest of the argument goes
  through as just written.
\end{proof}

\begin{lemma}\label{lem:ramseyconst}
  If an opt LCL $\Pi$ can be $\alpha$-approximated in $o(\log^\ast n)$ rounds in the
  \detlcl model, then $\alpha\bopt \geq \bconst$.
\end{lemma}

Our proof follows a similar idea of reference
\cite{czygrinow-hanckowiak-wawrzyniak-2008-fast-distributed} based on an
application of Ramsey's theorem.
For a set $X$ and $k \in \N_+$, let us write $\binom{X}{k}$ for the subsets of
$X$ of size exactly $k$.
For $k,r,s \in \N^+$ where $k \ge s \ge 2$, we let $R(k;r,s)$ denote the size of
the smallest $n \in \N^+$ such that every $r$-coloring $c\colon \binom{[n]}{s}
\to [r]$ of subsets of size $s$ of $[n]$ admits a subset $Y \in \binom{[n]}{k}$
for which $c$ is monochromatic on $\binom{Y}{s}$.
Such a number always exists by (the finite version of) Ramsey's theorem (see,
e.g., \cite{graham-rothschild-spencer-1991-ramsey-theory}) and can be
asymptotically upper-bounded \cite{erdos-rado-1952-combinatorial-theorems-on} by
a power tower of height $s$, that is,
\[
  R(k;r,s) = 2^{2^{\iddots^{2^{O(rk)}}}}
\]
where the number of occurrences of $2$ is equal to $s-1$.

The core of the argument is that, in the $o(\log^\ast n)$ regime, we can use
Ramsey's theorem to pick identifiers that force the output of the \detlcl
algorithm to be constant almost everywhere (i.e., on all but $o(n)$ nodes).
This fact is more or less folklore; we give a full proof for the sake of
completeness.

\begin{proof}
  It is sufficient to show that, for any \detlcl algorithm $\mathcal{A}$ with
  locality $T = o(\log^\ast n)$ and large enough $n$, there is an assignment of
  identifiers such that all but $o(n)$ many labels outputted by $\mathcal{A}$
  are identical.
  If that is the case, then we know that the value of the solution is such that
  $\alpha\bopt \ge \bconst - o(1)$.
  Since $\alpha$ is a constant, this implies $\alpha\bopt \ge \bconst$.

  Let us then prove the aforementioned statement.
  In preparation to apply Ramsey's theorem, we first set some parameters.
  (A reader who is not familiar with the argument may wish to ignore this
  paragraph for now and skip to the exposition below.)
  Set $r = \abs{\Gamma}$, $s = 2T + 1$, and $k = 2s^2$.
  In addition, let $m = (n - \sqrt{n})/k$.
  Observe that $n - mk = \sqrt{n}$ is at the same time $o(n)$ and larger than
  $R(k;r,s)$. 
  The latter holds because $r = O(1)$ and then
  \[
    k = O(\log^\ast n)^2 = o(\log\log^\ast n) = o(\log^s n)
  \]
  (where $\log^s$ indicates $s$-fold composition of $\log$ with itself), which
  together with the previous power tower estimate gives us $R(k;r,s) = O(\log
  n)$.

  Given the algorithm $\mathcal{A}$, we observe that it induces an $r$-coloring
  of $\binom{[n]}{s}$ as follows:
  Take a segment of the cycle with $s = 2T+1$ nodes $v_{-T},\dots,v_0,\dots,v_T$
  (arranged in this order) corresponding to the $T$-neighborhood of the node
  $v_0$.
  Given a set of $s$ identifiers $\mathfrak{i} = \{ i_{-T},\dots,i_T \} \in
  \binom{[n]}{s}$ where $i_{-T} < \dots < i_T$, place each identifier $i_j$ on
  node $v_j$ and consider the output of $\mathcal{A}$ at $v_0$.
  Since $\mathcal{A}$ is deterministic, the output is a fixed label
  $\ell_\mathfrak{i} \in \Gamma$.
  
  Applying Ramsey's theorem, we obtain that there is a set $I \in
  \binom{[n]}{k}$ with the property that, if we place the identifiers of $I$ in
  increasing order on a segment of $k$ nodes around the cycle, then the output
  of $\mathcal{A}$ inside this entire segment (i.e., on nodes that only see
  identifiers in $I$) is the \emph{same} label $\ell_I \in \Gamma$.

  Using our choice of parameters, we apply Ramsey's theorem not only once but
  recursively $m$ times in total (since $n-mk \ge R(k;r,s)$) to obtain a
  partition $I_0 + \dots + I_m = [n]$ where:
  \begin{enumerate}
    \item $\abs{I_0} = O(\sqrt{n})$ and $\abs{I_j} = k$ for $j \ge 1$.
    \item For every $I_j$ with $j \ge 1$, the output of $\mathcal{A}$ on the
    inside of an $I_j$ segment (in the preceding sense) is $\ell_{I_j}$.
  \end{enumerate}
  Additionally, since $r = \abs{\Gamma}$ is constant, we can assume without
  restriction that the output $\ell_{I_j}$ of $\mathcal{A}$ in every $I_j$ with
  $j \ge 1$ is not only fixed with respect to each $I_j$ but in fact an
  identical $\ell \in \Gamma$ across all the $I_j$'s.
  (E.g., we can blow up $n$ by a factor $r$ and then use an averaging argument.
  Since $r$ is constant, this does not interfere with the asymptotic estimates.)

  Finally, we construct a cycle of $n$ nodes by splitting it into $I_j$ segments
  as defined by the partition $I_0 + \dots + I_m$ and then placing identifiers
  in each $I_j$ segment in ascending order.
  By the previous reasoning, for every $I_j$ with $j \ge 1$, we know that the
  output of $\mathcal{A}$ is the same fixed label $\ell$ in all but $O(s)$ many
  nodes of $I_j$.
  (More specifically, these are the nodes at the boundaries between distinct
  segments and that see identifiers from $I_{j-1}$ and $I_j$ simultaneously.)
  Since there are $m$ such sets in total, there are only $\abs{I_0} + O(ms) =
  o(n)$ nodes that do not output $\ell$.
\end{proof}

\subsection{Upper Bounds}

Finally we show upper bound results, that is, how each of the classes of
\cref{table:complexclasses} corresponds to obtaining an approximation to a
desired factor $\alpha$.

\begin{lemma}\label{lem:bconstcomplexity}
  If $\alpha\bopt \geq \bconst$ holds for some min opt LCL $\Pi$, then the
  complexity class of an algorithm that finds an $\alpha$-approximation of $\Pi$
  is A.
\end{lemma}

\begin{proof}
$\alpha\bopt \geq \bconst$ implies there exists a constant solution that approximates the best possible solution for any $n$ sufficiently, as there is a self-loop in the de Bruijn graph. An algorithm finding the $\alpha$-approximation of $\Pi$ can output a constant solution in $O(1)$ in \detlcl and \randlcl. Together with \cref{lem:ramseyconst} this concludes the proof.
\end{proof}

\begin{lemma}
  Suppose $\alpha\bopt < \bconst$ for some min sum LCL $\Pi$ and additionally one of the following two conditions holds:
  \begin{enumerate}
    \item $\alpha\bopt \ge \bgap$ and $\dgap$ is false 
    \item $\alpha\bopt > \bgap$ and $\dgap$ is true 
  \end{enumerate}
  Then the complexity class of any algorithm finding the $\alpha$-approximation
  of $\Pi$ is B.
\end{lemma}

\begin{proof}
As the definitions of $\bgap$ and $\dgap$ are simply more extensive than $\bflex$ and $\dflex$, for \detlcl we can proceed like in \cref{lem:blfexcomplexity}. For \randlcl it is known that we can find a ruling set in $O(1)$ with high probability \cite{balliu-ghaffari-etal-2022-node-and-edge-averaged}. With a ruling set and $K$, like in \cref{lem:blfexcomplexity}, constant sequences can be labeled. Any super constant sequences can be labeled with the self-loop included in the definition of $\bgap$. With high probability the cycle will rarely resort to the cost of the self-loop even if it is higher than $\bgap$, so it is negligible on sufficiently large instances. Together with \cref{lem:symmbreak} this concludes the proof.
\end{proof}

The connection between the definition of $\bconst$ and a solution of an opt LCL in a directed cycle is comparably easily made. It is less obvious how the cheapest cycle in a flexible component is the deciding cost for finding labelings in complexity class C. Therefore we make a connection to \cref{sec:problem-param} with \cref{lem:flex-cycles} leading up to \cref{lem:blfexcomplexity}.

\begin{lemma}\label{lem:flex-cycles} 
  If $\bflex \neq \bot$ holds for a min sum LCL $\Pi$, then we can pick $K \geq K_0 \in \N_+$ so that the de Bruijn graph of $\Pi$ admits a closed
  walk $W$ of length $K$ for which $c(W) \leq \bflex + o(1)$.
\end{lemma}

\begin{proof}
Since $\bflex \neq \bot$ there exists at least one cycle through which $\bflex$ can be obtained. Let $C$ be one of these cycles. Because $C$ is in a flexible component $S$, we can find a closed walk $W'$ starting and ending in a node belonging to $C$, so there exists a $K'_0$ so that for all $K' \geq K_0':\text{ }\abs{W'} = K'$. This is achieved by connecting two closed coprime walks in $S$  to $C$. Now we compose a walk $W$ of length $K$ by finding a $K'$ so that $K-K'$ is a multiple of $C$. The amount of different walks we have to find to achieve this is clearly upper bounded by $C$. For large $K'$ (in summing problems) the ratio of the cost $c(C) = \bflex$ and $W'$ is improving. The worst cost ratio can be fixed by choosing $K_0$ large enough, where the potential additional cost impact of nodes belonging to $K'$ is $o(1)$ compared to the amount of nodes in $K$.
\end{proof}

\begin{lemma}\label{lem:blfexcomplexity}
  Suppose one of the following two conditions hold for some min sum LCL $\Pi$:
  \begin{enumerate}
    \item $\alpha\bopt \geq \bflex$ and $\dflex$ is false.
    \item $\alpha\bopt > \bflex$ and $\dflex$ is true.
  \end{enumerate}
  
  Then the complexity class of any algorithm finding the $\alpha$-approximation of $\Pi$ is C, assuming $\Pi$ does not belong to either class A or B.
\end{lemma}

\begin{proof}
The first of the two conditions from the lemma statement together with \cref{lem:flex-cycles} implies that there exists a constant $K$ for which any path of length $\geq K$
can be labeled in a way that this solution approximates the best possible solution for any $n$ sufficiently. A $3$-coloring can be found in $\Theta(\log^{*} n)$
\cite{cole-vishkin-1986-deterministic-coin-tossing-with,linial-1992-locality-in-distributed-graph-algorithms}.
Then this coloring can be used to find sequences of length at least $K$ which can then be labeled to produce a sufficient approximation in each sequence of the graph and therefore also in the whole graph. 
If the second condition from the claim holds, choose $n$ and therefore $K$ big enough so that $\alpha\bopt > \bflex$. Together with \cref{lem:flexexpect} this concludes the proof.
\end{proof}

\begin{lemma}
If $\alpha\bopt \geq \bcoprime$ for some min-max LCL $\Pi$ the complexity class of any algorithm finding the $\alpha$-approximation of $\Pi$ is C, assuming $\Pi$ does not belong to class A.
\end{lemma}

\begin{proof}
The proof is almost equal to the proof concerning the first condition in \cref{lem:blfexcomplexity}. Simply the way to arrive at the existence of the constant $K$ for which any path of length $\geq K$ can be labeled according to requirements is here directly rooting in the definition of $\bcoprime$. Together with \cref{lem:lowercoprime} this concludes the proof.
\end{proof}

\begin{lemma}
If $\bflex = \bot$ for a min-sum LCL $\Pi$ or $\bcoprime = \bot$ for a min-max LCL $\Pi'$ then it has no general solution and
its complexity class is E. If the $\alpha$-approximation algorithm is not in any
of the classes A, B, C or E, then it is in the class D.
\end{lemma}

\begin{proof}
If $\bflex = \bot$ (or equivalently $\bcoprime = \bot$), then the problem does not have a valid solution on all instances.
Otherwise, we know that $\alpha\bopt n < \bflex n + O(1)$ ($\alpha\bopt < \bcoprime$) holds for infinitely
many $n$, so on certain instances the whole graph must be known to all nodes in
order to obtain an $\alpha$-approximation.
\end{proof}

\bibliographystyle{plainurl}
\bibliography{da}

\appendix

\section{Connection of flexible components and walks of coprime lengths in digraphs}\label{apdix:flex-comprime}

\begin{lemma}\label{lem:2gamma-walks-flexibility}
Let $G=(V,A)$ be a strongly connected digraph of $n$ vertices. If $G$ has flexible vertices, then there exists a node $v\in V$, and two walks $A',B'$ containing $v$ such that $|A'|$ and $|B'|$ are coprime and of length $\leq 2n + 1$. 
\end{lemma}

\begin{proof}
  Suppose that $G$ is flexible. Let $v$ be a fixed node and $A$, $B$ be two closed walks containing $v$ such that $|A|$ and $|B|$ are coprime.  

Let $L =\{k \in \mathbb{N} \mid \exists W \in G \wedge |W| = k\}$ and $L' = L \cap [1, 2n+1]$. Per \cref{lem:coprime-cycles-equals-flexible-components}, it is clear that $\gcd(L)=1$ as $|A|$ and $|B|$ are coprime. 

We prove that $\gcd(L) = \gcd(L')$. Let $k_0 \in L\backslash L'$. We prove that $k_0$ can be expressed as a linear combination of strictly smaller $k_1,k_2$ and $k_3$. By repeating this argument on the $k_i$'s, we can express any $k_0$ as a linear combination of elements of $L'$, proving that $\gcd(L) = \gcd(L')$.  

Let $W \in G$ be a closed $v$ walk of length at least $2n + 2$. There must exist $v^* \in W$ such that $v^*$ is visited at least three times, or $v$ is visited four times.
We handle these two cases as follows:
\begin{enumerate}
  \item If $v$ is visited four times, we can split $W$ in three closed $v$ walks $w_1, w_2, w_3$ of lengths $k_1', k_2', k_3' \in [1, |W|-2]$. We then have $k_0 = k_1' + k_2' + k_3'$.
  \item Otherwise, suppose there exists $v^* \ne v$ that is visited at least
  three times. 
  Split $W$ into four walks at $v^*$, namely $w_1 = v, \dots, v^*$, $w_2 = v^*,
  \dots, v^*$,$w_3 = v^*, \dots, v^*$ and $w_4 = v^*, \dots, v$. Denote by
  $\ell_1,\ell_2,\ell_3$ and $\ell_4$ their respective lengths. Now consider the
  following $v$ closed walks: $X_1 = w_1 + w_4$, $X_2 = w_1 + w_2 + w_4$ and
  $X_3 = w_1 + w_3+ w_4$. Denote by $k_1, k_2$ and $k_3$ their respective
  lengths. Finally, observe that for all $i$ we have $1 \le k_i \le |W|-1$ and
  $k_0 = -k_1 + k_2 + k_3$.
\end{enumerate}
This concludes the proof that $\gcd(L) = \gcd(L')$. We then have that $1 = \gcd(L) = \gcd(L')$. Remark now that this proof is done in $G$, so there must exist two walks $A',B'$ of coprime lengths, lengths at most $2n+1$.
\end{proof}

\end{document}